\documentclass[a4paper,UKenglish,cleveref, autoref, thm-restate]{lipics-v2021}
\pdfoutput=1

\hideLIPIcs 

\title{Range Reporting for Time Series via Rectangle Stabbing} 

\author{Lotte Blank}{University of Bonn, Germany }{lblank@uni-bonn.de}{https://orcid.org/0000-0002-6410-8323}{}

\author{Anne Driemel}{University of Bonn, Germany}{driemel@cs.uni-bonn.de}{https://orcid.org/0000-0002-1943-2589}{}

\authorrunning{L. Blank and A. Driemel} 

\acknowledgements{This work was funded by  390685813 (Germany’s Excellence Strategy – EXC-2047/1: Hausdorff Center for Mathematics); 416767905; and the Deutsche Forschungsgemeinschaft (DFG, German ResearchFoundation) – 459420781 (FOR AlgoForGe)}

\Copyright{Lotte Blank and Anne Driemel} 



\category{} 

\relatedversion{} 

\nolinenumbers 

\EventEditors{John Q. Open and Joan R. Access}
\EventNoEds{2}
\EventLongTitle{42nd Conference on Very Important Topics (CVIT 2016)}
\EventShortTitle{CVIT 2016}
\EventAcronym{CVIT}
\EventYear{2016}
\EventDate{December 24--27, 2016}
\EventLocation{Little Whinging, United Kingdom}
\EventLogo{}
\SeriesVolume{42}
\ArticleNo{23}

\newtheorem{problem}[theorem]{Problem}

\newcommand{\F}{{Fréchet distance }}
\newcommand{\Fm}{{Fréchet distance}}
\newcommand{\FQ}{{Fréchet queries }}

\begin{document}

\maketitle

\begin{abstract}
We study the Fréchet queries problem. It is a data structure problem, where we are given a set $S$ of $n$ polygonal curves and a distance threshold $\rho$. The data structure should support queries with a polygonal curve $q$ for the elements of $S$, for which the continuous Fréchet distance to $q$ is at most $\rho$. Afshani and Driemel in 2018 studied this problem for two-dimensional polygonal curves and gave upper and lower bounds on the space-query time tradeoff. We study the case that the ambient space of the curves is one-dimensional and show an intimate connection to the well-studied rectangle stabbing problem. Here, we are given a set of hyperrectangles as input and a query with a point $q$ should return all input rectangles that contain this point. Using known data structures for rectangle stabbing or orthogonal range searching this directly leads to a data structure with $\mathcal{O}(n \log ^{t-1} n)$ storage and $\mathcal{O}(\log^{t-1} n+k)$ query time, where $k$ denotes the output size and $t$ can be chosen as the maximum number of vertices of either (a) the stored curves or (b) the query curves. 
The resulting bounds improve upon the bounds by Afshani and Driemel in both the storage and query time. In addition, we show that known lower bounds for rectangle stabbing and orthogonal range reporting with dimension parameter $d= \lfloor t/2 \rfloor$ can be applied to our problem via reduction. 
\end{abstract}

\section{Introduction}
The \emph{\F} is a popular measure of similarity of two curves $q$ and $s$ with broad application in many areas, including geographical information science \cite{Meulemans2014SimilarityMA, GIS_3, Wenk2006AddressingTN}, computational biology \cite{CompBio, Wylie2013ProteinCP}, image processing \cite{Alt2009, Sriraghavendra2007FrchetDB}, and quantum chemistry~\cite{ZHU2023619}. 
We focus on a data structuring problem which we refer to as the \emph{\FQ problem}. Here, in the preprocessing phase, we are given a set $S$ of $n$ polygonal curves and the distance threshold $\rho$. The task is to store this set in a data structure that can answer the following type of queries efficiently: For a polygonal curve $q$, output all curves in $S$ that have distance at most $\rho$ to $q$. 
Afshani and Driemel \cite{AfshaniDriemel18} studied this problem in 2018 for two-dimensional curves providing non-trivial upper bounds for the exact case. Recently, Cheng and  Huang \cite{Cheng2023SolvingFD} have generalized their approach for higher dimensions. Other works on variants of this problem have focused on the approximate setting~\cite{BDNP21,DEBERG2013747,DP21,DS17,FFK20}.

We focus on the exact setting and---following previous work \cite{BDNP21,DP21}---we restrict the ambient space of the curves to be 1-dimensional, that is, they are time series.
Time series appear in massive amounts in many different applications where they are used to track, e.g., the change over time in stock market value, vitality parameters of patients,  atmospheric conditions, such as temperature, the Earth’s population, and the hourly requests of a webpage. 

In the following, we refer to the number of vertices of a polygonal curve as the complexity of the curve and usually assume that this complexity is constant. We will see that this complexity roughly corresponds to the dimension of the problem when viewed as a rectangle stabbing or orthogonal range reporting problem.

\subparagraph*{Previous work: Exact setting}
Afshani and Driemel \cite{AfshaniDriemel18} propose a data structure based on multi-level partition trees for two-dimensional curves using semi-algebraic range searching. An essential ingredient to their work is the use of a finite number of predicates that retain sufficient information on the curves to solve the Fr\'echet queries problem within the partition tree framework. 
The resulting data structure for polygonal curves in the plane has size $\mathcal{O}\left(n (\log\log n)^{\mathcal{O}({t_s}^2)}\right)$ and uses query time in $\mathcal{O}\left(\sqrt{n}\cdot \log^{\mathcal{O}({t_s}^2)}n+k\right)$, where $t_s$ denotes the complexity of the input curves and $k$ the output size. The same technique can be applied to solve the problem for 1-dimensional curves using orthogonal range searching. In this case, their bounds reduce to $\mathcal{O}\left(n\left({\log n}/{\log \log n}\right)^{\mathcal{O}({t_s}^2)}\right)$ size and $\mathcal{O}\left(\log n \left({\log n}/{\log \log n}\right)^{\mathcal{O}({t_s}^2)}+k\right)$ query time.
For all time series $s$ of complexity $t_s$, an $\mathcal{O}({t_s}^2)$-dimensional point $p(s)$ is stored. The dimension of $p(s)$ is quadratic in $t_s$, because for every pair of vertices of $s$ values depending on both vertices are stored. These are used to evaluate the predicates mentioned above. We substantially simplify these predicates and show that this leads to improved bounds in the 1-dimensional case. 

Cheng and  Huang \cite{Cheng2023SolvingFD} used the same predicates as Afshani and Driemel to build a data structure for $d$-dimensional polygonal curves. They constructed a set of polynomials such that their signs encode the truth values of those predicates. This leads to a data structure of size in $\mathcal{O}(t_q t_s n)^{\mathcal{O}(d^4{t_q}^2\log(dt_q))}$ and query time in $\mathcal{O}((d t_q)^{O(1)}\log(t_q t_s n)+k)$, where $t_q$ denotes the complexity of the query time series.

Afshani and Driemel \cite{AfshaniDriemel18} also provide lower bounds in the pointer machine model. Using a volume argument, they show a lower bound stating that every data structure with $Q(n) +\mathcal{O}(k)$ query time, where $k$ is the output size, has to use roughly $\Omega((n/Q(n))^2)$ space in the worst case even if the query curves are just line segments or points for the discrete \Fm.

\subparagraph*{Previous work: Approximate setting}
Bringmann, Driemel, Nusser, and Psarros~\cite{BDNP21} study the setting of 1-dimensional curves.
Their work focuses on the $c$-approximate version of the near-neighbor ($c$-ANN) problem under the Fr\'echet distance. In this problem, only one of the curves in the query range needs to be reported and only if the query range is not empty. The approximation is defined with respect to the query radius.  Using a bucketing method, they construct a set of curves $S'$ depending on the input curves $S$, which are stored in a dictionary. They show that, given a query curve $q$, there must exist a curve in $S'$ very close to $q$ if there exists some curve in $S$ within distance $\rho$ of $q$. 
In this way, they constructed a $(1+\varepsilon)$-ANN data structure of size $n\cdot\mathcal{O}(t_s/(t_q\varepsilon))^{t_q}$ and query time in $\mathcal{O}(1)^{t_q}$. The query time reduces to $\mathcal{O}(t_q)$ with the same space bound for the $(2+\varepsilon)$-ANN data structure. Furthermore, a $(2+\varepsilon)$-ANN data structure with linear size $\mathcal{O}(t_s n)$ and query time in $\mathcal{O}(1/\varepsilon)^{t_q+2}$ is obtained. 
Their lower bounds show tightness of these bounds in several parameters, assuming the complexity of the curves depends on $n$ (i.e., it is not a constant). To this end, they consider the total time necessary to build the data structure and to answer $n$ queries. They show that, assuming the Orthogonal Vectors Hypothesis, a running time of $n \cdot (t_s/t_q)^{\Theta(t_q)}$ is necessary for any data structure that achieves an approximation factor $\alpha \in [1,2)$. 

The conditional lower bounds of Bringmann et al.~\cite{BDNP21} also apply to the exact setting, however, they assume the parameters $t_s$ and $t_q$ to be non-constant. In light of this, we focus on the setting where $t_s$ and $t_q$ are constants independent of $n$.

\subparagraph*{Our results}
Section~\ref{s: preliminaries} contains a formal definition of the data structure problem studied in this paper: the \FQ problem. Section~\ref{s:datastructure_techniques} contains the definition and known results for rectangle stabbing, as well as its dual problem, orthogonal range reporting. Our analysis shows an intimate connection to these classical problems studied in computational geometry as we use them for deriving both upper and lower bounds for the \FQ problem.
We start in Section~\ref{s: Lower bound} with a reduction showing that both rectangle stabbing and orthogonal range reporting in $d$ dimensions can be solved using a data structure  for the \FQ problem using curves of complexity $t=2d$. 

In Section \ref{predicates}, we review the known predicates of Afshani and Driemel which are used to test the Fr\'echet distance within the partition tree framework. 
Section~\ref{s: lemmas} contains our main lemmas for simplifying these predicates and introduces the new concept of forward and backward numbers. Here, we take advantage of the fact that the direction of each edge of a time series can only be orientated forward or backward with respect to the $x$-axis. 

The resulting data structures are presented in Section~\ref{s:Data structure}. We present two variants. Let $t_s$ be the complexity of the input and $t_q$ of the query and assume $t_s$ and $t_q$ are constant. The first data structure has size in $\mathcal{O}(n \log ^{t_q-2}n)$ and uses query time in $\mathcal{O}(\log^{t_q-1}n+k)$ and is independent of $t_s$.
The second data structure has size in $\mathcal{O}(n(\log n/\log \log n)^{t_s-1})$ and query time in $\mathcal{O}(\log n(\log n/\log \log n)^{t_s-3}+k)$ and is independent of $t_q$. In both variants, $k$ denotes the size of the output (without duplicates). 

Together with known lower bounds for rectangle stabbing and orthogonal range reporting, our analysis in Section~\ref{s: Lower bound} implies that every data structure that solves the \FQ problem and uses $nh$ space has to use $\Omega(\log n(\log n/\log h)^{\lfloor t/2\rfloor-2}+k)$ query time, where $t=\min\{t_q, t_s\}$.  If  the data structure uses query time in $\mathcal{O}(\log ^c n+k)$, where $c$ is a constant, it must use $\Omega(n(\log n/\log \log n)^{\lfloor t/2\rfloor-1})$ space.


\section{Problem Definition}\label{s: preliminaries}
For any two points $p, q\in \mathbb{R}^d$, $\overline{pq}$ is the directed line segment from $p$ to $q$. The linear interpolation of each pair of consecutive vertices of a sequence of vertices ${s_1,\dotso, s_{t_s}}\in \mathbb{R}^d$ is called a polygonal curve and its \emph{complexity} is the number of its vertices. This curve is also denoted as $\langle s_1, \dotso, s_{t_s}\rangle$. We can represent polygonal curves as functions $s:[1, t_s]\rightarrow \mathbb{R}^d$, where 
$s(i+\alpha)=(1-\alpha) s_i+\alpha s_{i+1}$ for $i\in \{1, \dotso, t_s\}$ and $\alpha \in [0, 1]$. The \emph{(continuous) \F} between polygonal curves $q: [1,t_q]\rightarrow \mathbb{R}^d$ and $s:[1,t_s]\rightarrow \mathbb{R}^d$ is defined as
\[ d_{\text{F}}(q, s)=\inf_{h_q\in \mathcal{F}_q, h_s\in \mathcal{F}_s }\ \max_{p\in [0,1]} \| q(h_q(p))-s(h_s(p))\|_2,\]
where $\mathcal{F}_q$ is the set of all continuous, non-decreasing functions $h_q: [0,1]\rightarrow [1,t_q]$ with $h_q(0)=1$ and $h_q(1)=t_q$, respectively $\mathcal{F}_s$ for $s$.

\begin{problem}[Fr\'echet queries]\label{prob}
Given a set $S$ of $n$ time series all of complexity at most $t_s$, the complexity $t_q$ of the query time series and a distance parameter $\rho \in \mathbb{R}_{\geq 0}$. Find a data structure that stores this set $S$ and can answer the following type of queries. For any query time series $q$ of complexity $t_q$, return all elements of $S$ that have continuous \F at most $\rho$ to $q$.
\end{problem}

\begin{example}\label{e:small-example}
We give a simple example demonstrating why our results are surprising. At first sight, it seems intriguing to believe that the set of queries corresponding to an input curve can be viewed as a finite union of axis-aligned hyperrectangles in the dimension of the (fixed) query curve complexity~$t_q$. However, a simple example shows that this is not always the case.
Let $s=\langle s_1, s_2\rangle$ be a time series of complexity $2$, where $s_1\leq s_2$. In \cref{s: lemmas}, we show the following statement. For every time series $q=\langle q_1, q_2, q_3, q_4\rangle$ of complexity~$4$, it holds that $d_F(q, s)\leq \rho$ if and only if $(q_1, q_2, q_3, q_4)$ is contained in the following set $W$ \[\{(p_1, p_2, p_3, p_4)|\ p_1\in [s_1-\rho, s_1+\rho], p_4\in [s_2-\rho, s_2+\rho], p_2, p_3\in [s_1-\rho, s_2+\rho], p_3\geq p_2-2\rho \}.\]

\begin{figure}
    \centering
    \includegraphics{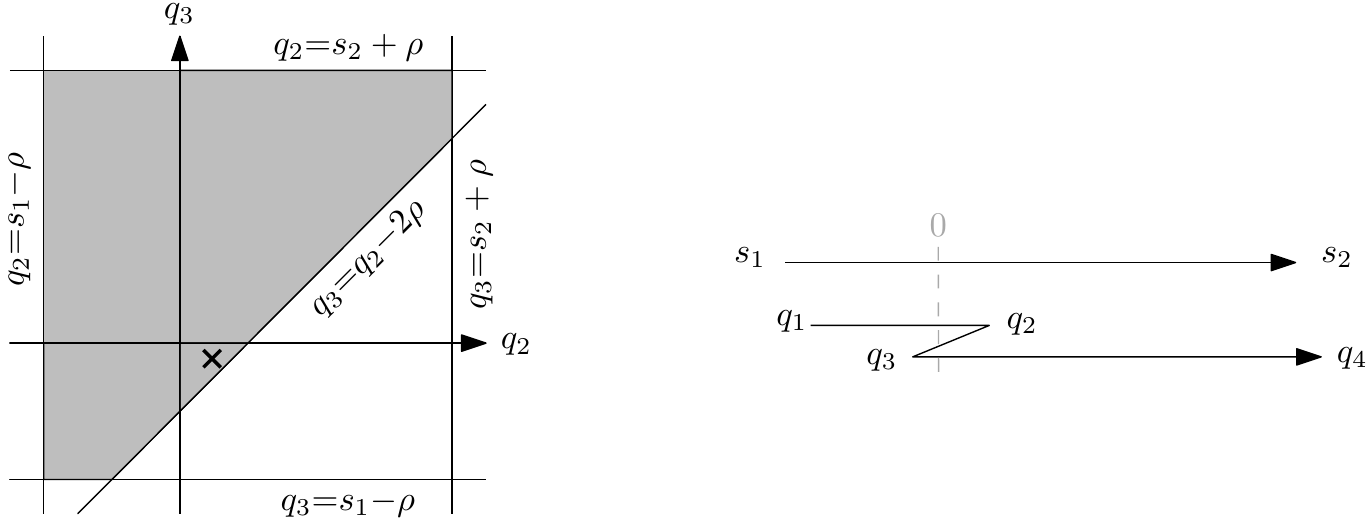}
    \caption{The $q_2$- and $q_3$-coordinates of the set $W$ of Example \ref{e:small-example}. Additionally, it must hold that $q_1\in [s_1-\rho, s_1+\rho]$ and $q_4\in [s_2-\rho, s_2+\rho]$. On the right is an example for such a time series $q$ with respect to $s$ and the corresponding point $(p_2, p_3)$ is marked.}
    \label{fig:small-example}
\end{figure}
The (non-orthogonal) condition $ p_3\geq p_2-2\rho$ stems from the monotonicity requirement in the definition of the Fr\'echet distance. 
The query space can be re-parameterized by introducing new variables to overcome this and to obtain a finite union of axis-aligned hyperrectangles, as this is implicitly done by Afshani and Driemel. To see this for this specific example, we can introduce an additional variable  $h$ with $h=q_2-q_3$ and $h \in [s_1-s_2-2\rho, 2\rho]$.
Achieving this with only a few additional variables (without blowing up the dimension quadratically as in the work of Afshani and Driemel) is the main challenge of our work. The key ingredient to our analysis is a simplification of the predicates---which goes along with a reduction of their overall number.
\end{example}

\section{Data Structure Techniques}\label{s:datastructure_techniques}

In this paper, we will show an intimate connection of the \FQ problem to rectangle stabbing and orthogonal range searching. We first describe these data structure problems independently and state the known results we will use in our analysis.

\subparagraph*{Rectangle Stabbing.} For rectangle stabbing the task is as follows. Preprocess a set $S$ of $n$ axis-aligned $d$-dimensional rectangles in $\mathbb{R}^d$ into a data structure so that all rectangles in $S$ containing a query point $q$ can be reported efficiently, ensuring that each such rectangle is reported exactly once.

Chazelle \cite{Cha86} developed a data structure for this problem with constant dimension $d$ that has size in $\mathcal{O}\left(n \log ^{d-2} n\right)$ and query time in $\mathcal{O}(\log^{d-1} n+k)$, where $k$ is the size of the output.
Afshani, Arge and Larsen~\cite{AAL12} proved the following lower bound for the rectangle stabbing problem. Any data structure that operates on a pointer machine and uses $n h$ space must use $\Omega\left(\log n(\log n/\log h)^{d-2}+k\right)$ query time, where $k$ is the output size. To prove a lower bound for the \FQ problem, we need a bounded version of rectangle stabbing. Here, all rectangles in $S$ are contained in $[0, 1]^d$.
The constructive proof for the lower bound uses only instances, where the input rectangles are all contained in a $d$-dimensional cube with side length $m<n$. By scaling this instance, we obtain that the lower bound holds also for bounded rectangle stabbing.

\subparagraph*{Orthogonal Range Searching.}
Orthogonal range searching is defined as follows.
Preprocess a set $S$ of $n$ points in $\mathbb{R}^d$ into a data structure so that for a $d$-dimensional axis-aligned query rectangle $R$ all points contained in $S$ can be reported efficiently, ensuring that each such point is reported exactly once.

Afshani, Arge and Larsen \cite{AAL12} construct a data structure for constant dimension $d>2$ that uses $\mathcal{O}\left(n(\log n/\log\log n)^{d-1}\right)$ space and query time in $\mathcal{O}\left(\log n(\log n/\log \log n)^{d-3}+k\right)$, where $k$ is the size of the output. Later, we reduce the orthogonal range searching problem to the \FQ problem and then use the following lower bound by Chazelle \cite{Chazelle_1990}. Consider a data structure of orthogonal range searching on $n$~points in $\mathbb{R}^d$ that operates on a pointer machine, and let $c$ be an arbitrary constant. If the data structure provides a query time of $\mathcal{O}((\log n)^c+k)$, where $k$ is the output size, then its size must be in $\Omega(n(\log n/\log \log n)^{d-1})$.




\section{Lower Bounds}\label{s: Lower bound}
We transform the bounded rectangle stabbing problem to the \FQ problem such that we can use a known lower bound for the bounded rectangle stabbing problem to obtain a lower bound for the \FQ problem. An illustration of the reduction can be found in Example \ref{e: reduction continuous}.

Given a set $S$ of $n$ axis-aligned rectangles contained in $[0,1]^d$ as an instance of the $d$-dimensional bounded rectangle stabbing problem. We define a set $S'$ containing $n$ time series of complexity $2d$. For a rectangle $R=\left[l_1, r_1\right]\times\left[l_2, r_2\right]\times\cdots \times\left[l_d, r_d\right]$ in $S$, we store the time series $s(R)=\langle s_1,\dotso, s_{2d}\rangle$, where  \[s_{2i-1}=(r_i+1)+6i \text{ and } s_{2i}=(l_i-1)+6i.\]
The set $S'$ is stored in a data structure for the \FQ problem. We define a query time series $q=\langle q_1, \dotso, q_{2d} \rangle$ for a query point $p=(p_1, \dotso, p_d)\in [0, 1]^d$, where \[q_{2i-1}=(p_i+2)+6i\text{ and } q_{2i}=(p_i-2)+6i.\]
To find all time series $s(R)$ in $S'$ within \F at most $1$ to $q$, we use the stored data structure. All rectangles $R$, where $d_{\text{F}}(q, s(R))\leq 1$, will be returned. Theorem \ref{t: reduction} implies that this reduction is correct by showing that
$p\in R\ \Leftrightarrow d_{\text{F}}(q, s(R))\leq 1$.

\begin{example}\label{e: reduction continuous}
    The input set $S$ of a rectangle stabbing instance contains the rectangles $R=[0.2, 0.6]\times[0.4, 1]\times[0.4, 0.6]$ and $\widehat{R}=[0, 0.4]\times[0.2, 0.6]\times[0.8, 1]$ and the query point is $p=(0.3, 0.8, 0.5)$. It is evident that $p\in R$ and $p\notin \widehat{R}$. Through our reduction, the two stored time series are ${s(R)=\langle 7.6, 5.2, 14, 11.4, 19.6, 17.4\rangle}$ and $s(\widehat{R})=\langle 7.4, 5, 13.6, 11.2, 20, 17.8\rangle$ and the query time series for $p$ is $q=\langle 8.3, 4.3, 14.8, 10.8, 20.5, 16.5 \rangle$, as illustrated in Figure \ref{fig: reduction continuous}.
    \begin{figure}
        \begin{center}
            \includegraphics[page=1]{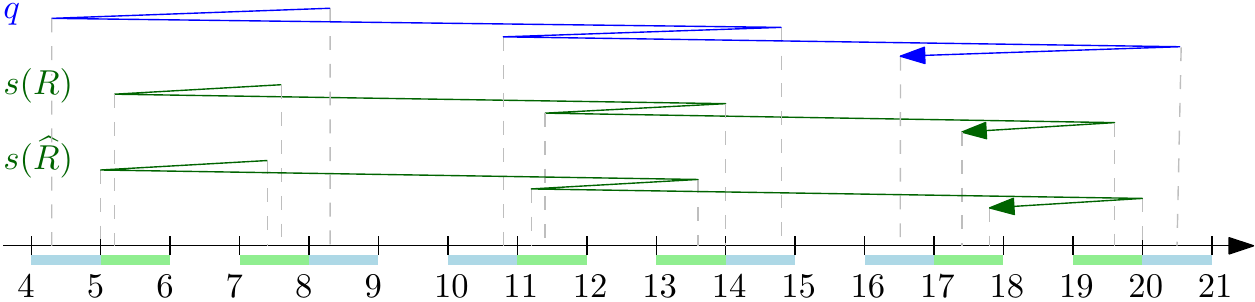}
            \caption{The time series $q$, $s(R)$ and $s(\widehat{R})$ as in Example \ref{e: reduction continuous}.}
            
            \label{fig: reduction continuous}
        \end{center}
    \end{figure}
    The left side of {Figure \ref{fig: reduction free space diagram}} depicts the free space diagram of the time series $q$ and $s(R)$ with respect to $\rho=1$ (i.e., $F_1(q, s(R))$). Notably, the points $(i,i)$ lie in the free space for all $i$, resulting in $d_{\text{F}}(q, s(R))\leq 1$. Conversely, the right side of {Figure \ref{fig: reduction free space diagram}} corresponds to the free space diagram of the time series $q$ and $s(\widehat{R})$ with respect to $\rho=1$ (i.e., $F_1(q, s(\widehat{R}))$). It does not contain a feasible path and $|q_3-s(\widehat{R})_3|>1$ and $|q_6-s(\widehat{R})_6|>1$. 
    \begin{figure}
        \begin{center}
            \includegraphics[page=1]{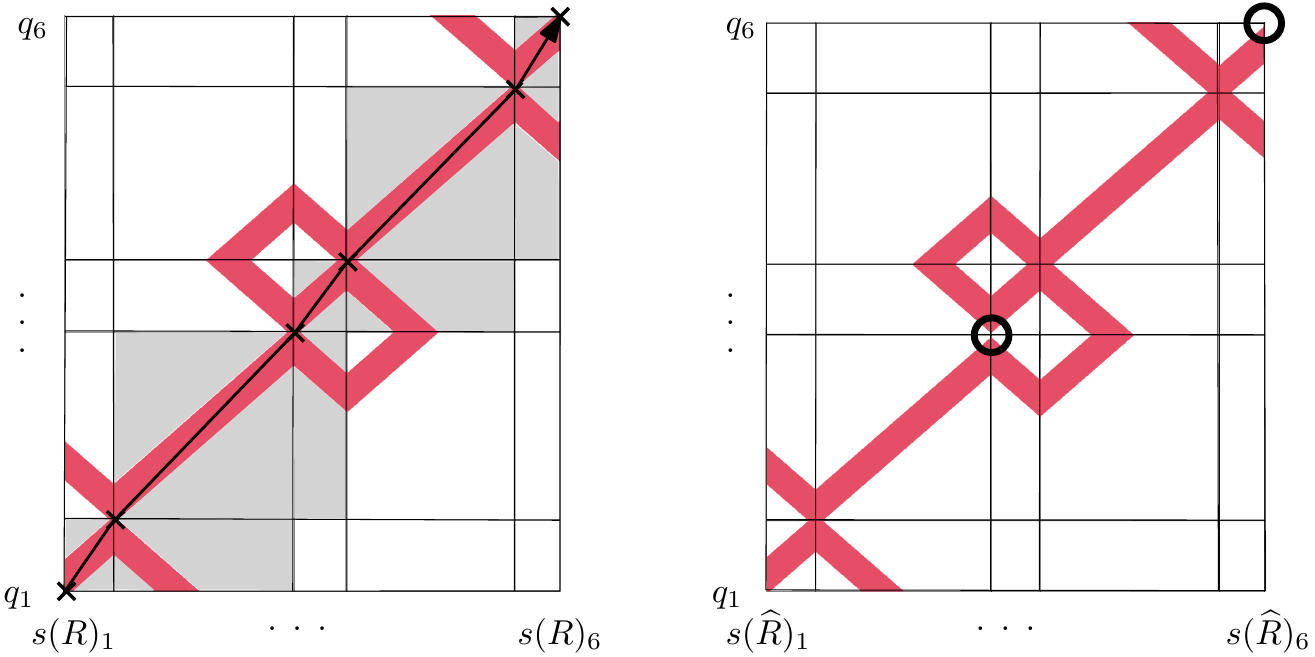}
            \caption{The free space diagrams $F_{1}(q, s(R))$ and $F_{1}(q, s(\widehat{R}))$ defined in Example \ref{e: reduction continuous}. A sequence of cells $\mathscr{C}$ that is feasible in $F_{1}(q,s(R))$ is drawn in grey.}
            \label{fig: reduction free space diagram}
        \end{center}
    \end{figure}
\end{example}

\begin{theorem}\label{t: reduction} 
    The $d$-dimensional bounded rectangle stabbing problem can be solved with a data structure for the \FQ problem, where the stored time series as well as the query time series have complexity $2d$. The instance for the \FQ problem can be computed in linear time.
\end{theorem}

\begin{proof}
    We use the reduction as described above. Hence, it remains to prove that
    $p\in R\ \Leftrightarrow d_{\text{F}}(q, s(R))\leq 1$.
    For all $i$, it follows by $l_i, r_i, p_i\in [0,1]$ that
    \begin{align*}
    &|q_{2i-1}-s_{2i-1}|\leq 1\ \Leftrightarrow\ |((p_i+ 2)+6i)-((r_i+1)+6i)|\leq 1\ \Leftrightarrow\  p_i\leq r_i \text{ and }\\
    &|q_{2i}-s_{2i}|\leq 1\ \Leftrightarrow\ |((p_i- 2)+6i)-((l_i-1)+6i)|\leq 1\ \Leftrightarrow\  p_i\geq l_i.
    \end{align*}
    If $p\in R$, then it holds that $|q_{2i-1}-s_{2i-1}|\leq 1$ and  ${|q_{2i}-s_{2i}|\leq 1}$ for all $i$, since $l_i\leq p_i\leq r_i$. Therefore, $\mathscr{C}=((1,1), (1, 2), (2,2), (2,3), (3, 3), (3, 4), \dotso, (2d-1, 2d), (2d, 2d))$ is a feasible sequence of cells in $F_1(q, s(R))$ because all cells are convex and boundary points of a cell belong to all neighboring cells. So, $d_{\text{F}}(q, s(R))\leq 1$.
    
    If $d_{\text{F}}(q, s)\leq 1$, then by the definition of the \Fm, for all points $q_{2i-1}$ and $q_{2i}$, there exist points $x_{2i-1}$ and $x_{2i}$ such that ${|q_{2i-1}-x_{2i-1}|\leq 1}$, $|q_{2i}-x_{2i}|\leq 1$, and $x_{2i-1}$ lies not after $x_{2i}$ on the time series $s(R)$. 
    By construction, it holds $s_{2k}-1<s_{2k-1}-1<q_{2i}<q_{2i-1}<s_{2l}+1<s_{2l-1}+1$ for all $k<i<l$.
    It holds that $|s_k-q_l|>1$ for $k\neq l$ by definition.
    Therefore, $x_{2i-1}, x_{2i}$ must lie on one of the following edges $\overline{s_{2i-2} s_{2i-1}}$, $\overline{s_{2i-1} s_{2i}}$, or $\overline{s_{2i} s_{2i+1}}$.
    By construction, it holds that $q_{2i-1}\in [6i+2, 6i+3]$ and $|x_{2i-1}-q_{2i-1}|\leq 1$. Hence, ${x_{2i-1}\in [6i+1, 6i+4]}$. Further, $q_{2i} \in [6i-2, 6i-1]$ and $|x_{2i}-q_{2i}|\leq 1$. Hence, $x_{2i}\in  [6i-3, 6i]$. Assume that $x_{2i-1}\in \overline{s_{2i} s_{2i+1}}$. Then since $x_{2i}$ lies after $x_{2i-1}$ on $s(R)$ it follows that $x_{2i}\in \overline{s_{2i} s_{2i+1}}$ and in particular $6i+1\leq x_{2i-1}\leq x_{2i}$. This leads to a contradiction to $x_{2i}\leq 6i$. In the same way, it follows that $x_{2i}\notin \overline{s_{2i-2} s_{2i-1}}$. So, $x_{2i-1}$ lies on $\overline{s_{2i-2} s_{2i-1}}$ or $\overline{s_{2i-1} s_{2i}}$ and $x_{2i}$ lies on $\overline{s_{2i-1} s_{2i}}$ or $\overline{s_{2i} s_{2i+1}}$. It holds that  $x_{2i-1}\leq s_{2i-1}$ and $x_{2i}\geq s_{2i}$, because $s_{2i-2}\leq s_{2i}\leq s_{2i-1}\leq s_{2i+1}.$
    It follows by $x_{2i-1}\leq s_{2i-1}\leq q_{2i-1}$ and $|x_{2i-1}-q_{2i-1}|\leq 1$ that
    $|q_{2i-1}-s_{2i-1}|= q_{2i-1}-s_{2i-1}\leq  q_{2i-1}-x_{2i-1}=|q_{2i-1}-x_{2i-1}| \leq 1.$ By the same argument, it follows that $|s_{2i}-q_{2i}|\leq 1$ because $q_{2i}\leq s_{2i}\leq x_{2i}$. 
    Therefore, $ l_i\leq p_i\leq r_i$ for all $i$, i.e., $p\in R$, which concludes the proof.
\end{proof}

The result in Theorem~\ref{t: reduction} together with the lower bound for bounded rectangle stabbing queries by Afshani, Arge and Larsern~\cite{AAL12} yields the following lower bound for the \FQ problem. 

\begin{corollary}
    Every data structure that solves the \FQ problem that operates on a pointer machine, and uses $n h$ space must use query time in ${\Omega(\log n(\log n/\log h)^{\lfloor t/2\rfloor -2}+k)}$, where $k$ is the size of the output (without duplicates) and $t=\min\{t_q, t_s\}$.
\end{corollary}

Given an instance of $d$-dimensional orthogonal range searching, we can construct the stored (resp. query) time series in the way as the query (resp. stored) time series were constructed in Theorem \ref{t: reduction} after scaling the instance such that all points are in $[0,1]^d$. Using this construction, it holds by the same arguments as in the proof of Theorem \ref{t: reduction} that $p\in R$ if and only if $d_F(q, s(p))\leq 1$. Therefore, we get the following corollary.

\begin{corollary}
    The $d$-dimensional orthogonal range searching can be solved with a data structure for the \FQ problem, where the stored time series as well as the query time series have complexity $2d$. The instance for the \FQ problem can be computed in linear time.
\end{corollary}

Chazelles \cite{Chazelle_1990} lower bound for orthogonal range searching provides to the following: 

\begin{corollary}
    Every data structure that solves the \FQ problem and uses ${\mathcal{O}(\log^c n+k)}$ query time, where $c$ is a constant, must use $\Omega(n(\log n/\log\log n)^{\lfloor t/2\rfloor-1})$ space, where $k$ is the size of the output (without duplicates) and $t=\min\{t_q, t_s\}$.
\end{corollary}

\section{Predicates for Evaluating the \Fm}\label{predicates}

In this section, we review the predicates used by Afshani and Driemel and how they enable the evaluation of the Fréchet distance in a data structure context. 

For this, we first recall the definition of the free space diagram from Alt and Godau \cite{CompFreDist}.
For polygonal curves $q: [1, t_q]\rightarrow \mathbb{R}^d$ and $s: [1, t_s]\rightarrow \mathbb{R}^d$ the \emph{free space diagram} $F_\rho(q, s)$ is a subset of $[1, t_q]\times[1, t_s]$, such that for all points $(x, y)\in F_\rho(q, s)$ the distance between $q(x)$ and $s(y)$ is at most~$\rho$. Refer to Figure \ref{fig:free space diagram} for an example. Formally,
\[F_\rho(q, s):=\{(x, y)\in [1, t_q]\times[1, t_s]\ |\ \|q(x)-s(y)\|_2\leq \rho\}.\]

\begin{lemma}[Alt and Godau \cite{CompFreDist}]\label{lAltGodau}
    The \F between polygonal curves $s$ and $q$ is at most $\rho$ if and only if there exists a path in $F_\rho(q, s)$ from $(1,1)$ to $(t_q, t_s)$ which is monotone in both coordinates. For such a path, we say it is  \emph{feasible}.
\end{lemma}
We can decompose  the rectangle $[1, t_q]\times[1, t_s]$ into $(t_q-1)\cdot (t_s-1)$ cells such that the cell $C_{i j}=[i, i+1]\times [j, j+1]$ corresponds to the part in the free space diagram defined by the edges $\overline{q_i q_{i+1}}$ and $\overline{s_j s_{j+1}}$. By definition of the free space diagram, it follows that $C_{i j}\cap F_{\rho}(q, s)$ lies between two parallel lines. Therefore, we focus on the boundary of the cells $C_{i j}$.

\begin{figure}
    \centering
    \includegraphics{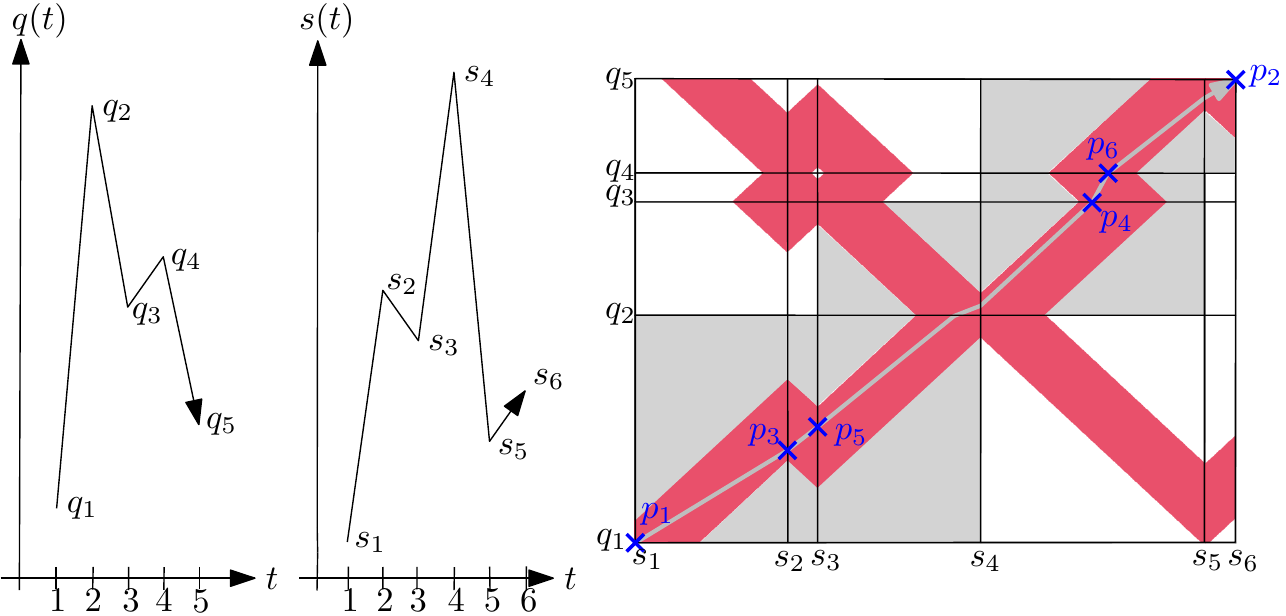}
    \caption{The free space diagram $F_{\rho}(q,s)$ of two time series with a feasible path trough a feasible sequence of cells $\mathscr C=((1,1), (1,2), (1, 3), (2, 3), (2, 4), (3, 4), (4, 4), (4, 5))$, which is drawn in grey. Predicates $(P_1), (P_2), (P_3(1,2)), (P_4(3,4)), (P_5(1,2,3))$ and $(P_6(3,4,4))$ are true, because the points $p_i$ are contained in the free space.}
    \label{fig:free space diagram}
\end{figure}

Our query algorithm will iterate over all possibilities of sequences of cells that a feasible path could traverse in the free space diagram. Therefore, we call a sequence of cells $\mathscr{C}=((i_1,j_1), \dotso,(i_t, j_t))$ \emph{valid}, if $i_1=1$, $j_1=1$, $i_k=t_q-1$, $j_k=t_s-1$, and for all $m<t$ either $i_{m+1}=i_{m}$ and $j_{m+1}=j_{m}+1$, or $i_{m+1}=i_{m}+1$ and $j_{m+1}=j_{m}$. The tuple~$(i, j)$ represents the cell $C_{i j}$. Further, a valid sequence of cells is called \emph{feasible} in $F_\rho(q, s)$, if there exists a feasible path in $F_\rho(q, s)$ that traverses exactly the cells in~$\mathscr{C}$. Refer to \cref{fig:free space diagram} for an example.

The following predicates due to Afshani and Driemel \cite{AfshaniDriemel18}  can be used to decide whether a valid sequence of cells is feasible in $F_\rho(q, s)$. \Cref{fig:free space diagram} visualizes the predicates. 

\begin{description}
    \item[$(P_1)$] \emph{(Endpoint (start))} This predicate is true iff $|s_1-q_1|\leq \rho$.
    \item[$(P_2)$] \emph{(Endpoint (end))} This predicate is true iff $|s_{t_s}-q_{t_q}|\leq \rho$.
    \item[$(P_3 (i, j))$] \emph{(Vertex of $s$ - edge of $q$)} This predicate is true iff $\exists\ p_3\in \overline{q_i q_{i+1}}$ s.t. $|p_3-s_j|\leq \rho$.
    \item[$(P_4 (i, j))$] \emph{(Vertex of $q$ - edge of $s$)} This predicate is true iff $\exists\ p_4\in \overline{s_j s_{j+1}}$ s.t. $|p_4-q_i|\leq \rho$.
    \item[$(P_5 (i, j, k))$] \emph{(Monotone in $q$)} This predicate is true iff $\exists\ p_3,p_5\in \overline{q_i q_{i+1}}$ s.t. $p_3$ lies not after~$p_5$ on the time series $q$ and $|p_3-s_j|\leq\rho$ and $|p_5-s_k|\leq\rho$.
    \item[$(P_6(i, l, j))$] \emph{(Monotone in $s$)} This predicate is true iff $\exists\ p_4,p_6\in \overline{s_j s_{j+1}}$ s.t. $p_4$ lies not after~$p_6$ on the time series $s$ and $|p_4-q_i|\leq\rho$ and $|p_6-q_l|\leq\rho$.
\end{description}


The following lemma verifies that the predicates can be used to test if the Fréchet distance between two curves is at most a given value. 

\begin{lemma}[Afshani and Driemel \cite{AfshaniDriemel18}] \label{lsequenceofcells}
    Let $\mathscr{C}=((i_1,j_1), (i_2,j_2),\dotso,(i_t, j_t))$ be a valid sequence of cells. Then $\mathscr{C}$ is feasible in $F_{\rho}(q, s)$ if and only if the following predicates defined by $q$, $s$ and $\rho$ are true:
    \begin{enumerate}
        \item[(i)] $(P_1)$ and $(P_2)$,
        \item[(ii)] $(P_3(i, j))$ if $(i, j-1),(i, j)\in \mathscr{C}$,
        \item[(iii)] $(P_4(i, j))$ if $(i-1, j),(i, j)\in \mathscr{C}$,
        \item[(iv)] $(P_5(i, j, k))$ if $(i, j-1), (i, k)\in \mathscr{C}$, and
        \item[(v)] $(P_6(i, l, j))$ if $(i-1, j), (l, j)\in \mathscr{C}$.
    \end{enumerate}
    We say that those predicates are induced by $\mathscr{C}$. 
\end{lemma}

Afshani and Driemel show that, for a given query, a fixed assignment of truth values to the set of all predicates defines a semi-algebraic set. This set contains all curves for which the predicates yield the given truth assignment. A query to the data structure then corresponds to a finite union of semi-algebraic range queries for which the truth assignments yield a valid sequence of cells. 

In our paper, we modify this approach. Instead of fixing the truth assignment to all predicates, we only fix a combinatorial path in the free space diagram (that is a valid sequence of cells) and we consider the predicates that are induced by it. This results in potential duplicates in the query output as an input curve may have different combinatorial paths in the free space diagram with the query. However, the overall number of elements in the output only changes by a constant factor as long as the complexity of the input and query is constant.

\section{Simplification of the Predicates}\label{s: lemmas}

Given a sequence of cells $\mathscr C$, we want to find intervals $I_1, \dotso, I_{t_q}$ defined by a stored time series $s$ such that $\mathscr C$ is feasible in $F_\rho(q,s)$ if and only if $q_i\in I_i$ for all $i$, where $q=\langle q_1, \dotso, q_{t_q}\rangle$ is a time series with some additional properties. The intervals will be defined using the predicates. Lemma \ref{lsequenceofcells} shows which predicates need to be true such that $\mathscr C$ is feasible in $F_\rho(q,s)$. For the endpoint and vertex-edge predicates ($(P_1), (P_2), (P_3)$ and $(P_4)$), the needed intervals follow easily:

\begin{observation}\label{OP1234}
    Let $q=\langle q_1,\dotso, q_{t_q}\rangle$ and $s=\langle s_1,\dotso, s_{t_s}\rangle$ be two time series. Then the following holds for the predicates in the free space diagram $F_{\rho}(q, s)$:
    \begin{enumerate}
        \item [{(i)}] $(P_1)\text{ is true }\Leftrightarrow q_1\in [ s_1-\rho, s_1+\rho]$,
        \item [{(ii)}] $(P_2)\text{ is true } \Leftrightarrow q_{t_q}\in [s_{t_s}-\rho, s_{t_s}+\rho]$,
        \item [{(iii)}] $(P_3(i, j)) \text{ is true } \Leftrightarrow \text{ if } q_i\leq q_{i+1}:\ q_i\leq s_j+\rho \text{ and } q_{i+1}\geq s_j -\rho$ and \\ \hspace*{3.2 cm} if $q_i\geq q_{i+1}: \ q_i\geq s_j -\rho \text{ and } q_{i+1}\leq s_j+\rho$,
        \item [{(iv)}] $(P_4(i, j)) \text{ is true }  \Leftrightarrow q_i\in [\min \{s_j-\rho, s_{j+1}-\rho\},\max\{s_j+\rho, s_{j+1}+\rho\}]$.
    \end{enumerate}
\end{observation}

The next lemma defines the intervals needed such that the monotone in $q$ predicate $(P_5(i, j, k))$ is true and is visualized in Figure \ref{fig:predicateP5}. 

\begin{figure}\centering
    \includegraphics[page=1]{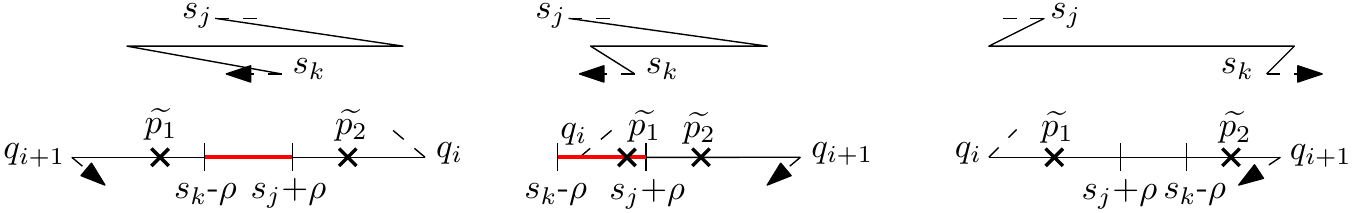}
    \caption{The left and the middle show Case (i) of Lemma \ref{lP5}. Here, $[s_j-\rho, s_j+\rho]\cap[s_k-\rho, s_k+\rho]$ is marked in red. The right visualizes Case (ii) of Lemma \ref{lP5} and Case (iii) is Case (ii) mirrored.}
    \label{fig:predicateP5}
\end{figure}

\begin{lemma}\label{lP5}
    Let $q=\langle q_1,\dotso, q_{t_q}\rangle$ and $s=\langle s_1,\dotso, s_{t_s}\rangle$ be two time series, then the monotone in $q$ predicate $(P_5(i, j, k))$ is true if and only if the vertex of $s$ - edge of $q$ predicates $(P_3(i, j))$ and $(P_3(i, k))$ are true and one of the following holds
        \begin{enumerate}
            \item[(i)] $ |s_j-s_k|\leq 2\rho$, or
            \item [(ii)] $|s_j-s_k|>2\rho$ and $s_j\leq s_k$ and $q_i\leq s_j+\rho$ and $q_{i+1}\geq s_k-\rho$, or
            \item [(iii)] $|s_j-s_k|>2\rho$ and $s_j> s_k$ and $q_{i}\geq s_j-\rho$ and $q_{i+1}\leq s_k+\rho$.
        \end{enumerate}
\end{lemma}

\begin{proof}
    Assume $(P_5(i, j, k))$ to be true. Then there exist points $p_1, p_2\in \overline{q_{i} q_{i+1}}$ such that $p_1$ lies not after $p_2$ on the time series $q$ and $p_1\in [s_j-\rho, s_j+\rho]$, $p_2\in [s_k-\rho, s_k+\rho]$. Hence, $(P_3(i,j))$ and $(P_3(i, k))$ are true. In addition, if $|s_j-s_k|>2\rho$ and $s_j\leq s_k$, it holds that $q_i\leq p_1\leq s_j+\rho<s_k-\rho\leq p_2\leq q_{i+1}$. Further, if $|s_j-s_k|>2\rho$ and $s_j> s_k$, it holds that $q_{i+1}\leq p_2\leq s_k+\rho<s_j-\rho\leq p_1\leq q_{i}$.
    
    It remains to prove the other direction. Let $(P_3(i, j))$ and $(P_3(i, k))$ be true. Then, there exist $\widetilde{p_1}\in \overline{q_i q_{i+1}}$ and $\widetilde{p_2}\in \overline{q_i q_{i+1}}$ such that $|\widetilde{p_1}-s_j|\leq\rho$ and $|\widetilde{p_2}-s_k|\leq\rho$.\\
    {Case $(i)$:} Let $|s_j-s_k|\leq2\rho$. We can set $p_1=p_2=\widetilde{p_1}$ if ${\widetilde{p_1}\in[s_j-\rho, s_j+\rho]\cap[s_k-\rho, s_k+\rho]}$. The same holds if $\widetilde{p_2}\in[s_j-\rho, s_j+\rho]\cap[s_k-\rho, s_k+\rho]$. Otherwise, $\emptyset\neq[s_j-\rho, s_j+\rho]\cap[s_k-\rho, s_k+\rho]\subseteq \overline{\widetilde{p_1}\widetilde{p_2}}\subseteq \overline{q_i q_{i+1}}$ and we can set $p_1=p_2$ as any point in this intersection. In each of those cases, $(P_5(i, j, k))$ is true. \\   
    {Case $(ii)$:} Let $|s_j-s_k|>2\rho$, $s_j\leq s_k$, $q_i\leq s_j+\rho$ and $q_{i+1}\geq s_k-\rho$. Then it holds that $\widetilde{p_1}<\widetilde{p_2}$ and $q_i<q_{i+1}$ because $s_j+\rho<s_k-\rho$. Therefore, $\widetilde{p_1}$ lies before $\widetilde{p_2}$ on $\overline{q_i q_{i+1}}$ and we can simply set $p_1=\widetilde{p_1}$ and $p_2=\widetilde{p_2}$. Those points have the required properties in the definition of $(P_5(i, j, k))$. Symmetrically, in Case $(iii)$ it holds that $(P_5(i, j, k))$ is true.
\end{proof}
To determine the truth value of the monotone in $s$ predicates $(P_6)$, we define the \emph{forward} and \emph{backward numbers} $f_i(q)$ and $b_i(q)$. Refer to Figure \ref{fig:f and b} as an example.

\begin{definition}[forward and backward numbers]\label{d: forward number}
    For a time series $q= \langle q_1,\dotso, q_{t_q}\rangle$ and $i\in~\{1, \dotso, t_q\}$, we denote by the \emph{forward number} $f_i(q)$ the highest number such that $\langle q_i-\rho, q_{f_i(q)}+\rho\rangle$ is oriented forward and its \F to the time series $\langle q_i,\dotso, q_{f_i(q)}\rangle$ is at most $\rho$, i.e., 
    \[f_i(q):= \max\{k\in \{i,\dotso, t_q\} \ \ |\ d_{\text{F}}(\langle q_i,\dotso, q_k\rangle, \langle q_i-\rho, q_k+\rho\rangle)\leq \rho \text{ and } q_i-\rho\leq q_k+\rho\}\]
    and by the \emph{backward number} $b_i(q)$ the highest number such that $\langle q_i+\rho, q_{b_i(q)}-\rho\rangle$ is oriented backward and its \F to the time series $\langle q_i,\dotso, q_{b_i(q)}\rangle$ is at most $\rho$, i.e., 
    \[b_i(q):= \max\{k\in \{i,\dotso, t_q\} \ \ |\ d_{\text{F}}(\langle q_i,\dotso, q_k\rangle, \langle q_i+\rho, q_k-\rho\rangle)\leq \rho  \text{ and } q_i+\rho\geq q_k-\rho\}.\]
\end{definition}


\begin{figure}
    \begin{center}
        \includegraphics[page=12]{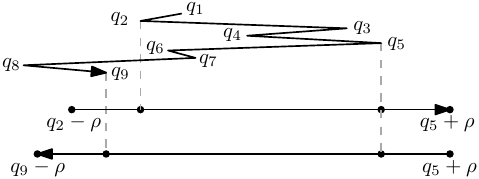}
    \end{center}
    \caption{Illustration of the values $f_2(s)=5$ and $b_5(s)=9$ for a time series $s$.} 
    \label{fig:f and b}
\end{figure}

\begin{observation} \label{o: f and b}
    For all $ i\leq x\leq  f_i(q)$, it holds that $d_{\text{F}}(\langle q_i, \dotso, q_x\rangle, \langle q_i-\rho, q_x+\rho\rangle)\leq \rho$ and $q_i-\rho\leq q_x+\rho$. Respectively, for $b_i(q)$.
\end{observation}
\begin{proof}
    By the definition of the \F and $f_i(q)$, there exist points $q_i-\rho\leq p_i\leq p_{i+1}\leq \dotso\leq p_x\leq q_x+\rho$ such that $|p_j-q_j|\leq \rho$ for all $j$. Further, since the free space in every cell is convex, the statement follows.
\end{proof}
 
The next lemma shows how the forward and backward numbers can be used to determine values of the monotone in $s$ predicates $(P_6)$. To decide whether a valid sequence of cells is feasible or not in $F_\rho(q,s)$, we need predicate $(P_6(i, l, j))$ to be true only if we also need all predicates $(P_6(x, y, j))$ to be true with $i\leq x < y\leq l$ by \cref{lsequenceofcells}. 

\begin{lemma}\label{lP6}
    Let $q=\langle q_1,\dotso, q_{t_q}\rangle$ and $s=\langle s_1,\dotso, s_{t_s}\rangle $ be two time series, ${i, l\in \{1,\dotso,t_q\}}$ with $i<l$ and $j\in \{1,\dotso,t_s-1\}$.
    If $s_j\leq s_{j+1}$, then \[(P_6(x, y, j))\text{ is true }\forall i\leq x< y\leq l\   \Leftrightarrow\  f_i(q)\geq l \text{ and\ }(P_4(x, j))\text{ is true for all } i\leq x\leq l\] and if $s_j\geq s_{j+1}$, then \[(P_6(x, y, j))\text{ is true }\forall i\leq x < y \leq l\  \Leftrightarrow\  b_i(q)\geq l \text{ and\ }(P_4(x, j))\text{ is true for all }  i\leq x\leq l.\]
\end{lemma}

\begin{proof}
    We discuss the case that $s_j\leq s_{j+1}$. The other case can be proven in the same way. 
    
    Let $(P_6(x, y, j))$ be true for all $i\leq x<y\leq l$. Note that by definition of $(P_6)$, it holds that the predicates $(P_4(x, j))$ are true for all $i\leq x\leq l$. It remains to prove $f_i(q)\geq l$.
    Let $i\leq x<y\leq l$. Since $(P_6(i, x, j))$ is true there exist $p_4,p_6\in \overline{s_j s_{j+1}}$ such that $p_4\leq p_6$ and $p_i\leq p_4+\rho$ and $q_x\geq p_6-\rho$. Therefore, it follows that $q_i-q_x\leq p_4+\rho-(p_6-\rho)\leq 2\rho$. In the same way, we get that $q_x-q_l\leq2\rho$. In particular, this implies that $\langle q_i-\rho, q_l+\rho\rangle$ is a forward edge. Assume for the sake of a contradiction that $d_{\text{F}}(\langle q_i, \dotso, q_l\rangle, \langle q_i-\rho, q_l+\rho\rangle)> \rho$. Then, there must exist two vertices $q_x, q_y$ such that there are no two points $ p_1, p_2\in [q_i-\rho, q_l+\rho]$ such that $p_1\leq p_2$, $|p_1-q_x|\leq \rho$ and $|p_2-q_y|\leq \rho$. Further, since $(P_6(x, y, j))$ is true we know that there exist $p_4\leq p_6$ on $\overline{s_j s_{j+1}}$ such that $|p_4-q_x|\leq \rho$ and $|p_6-q_y|\leq \rho$. Now set $\widetilde{p_1}=\max\{\min\{q_l+\rho, p_4\}, q_i-\rho\}$ and $\widetilde{p_2}=\max\{\min\{q_l+\rho, p_6\}, q_i-\rho\}$. It holds that $\widetilde{p_1},\widetilde{p_2}\in [q_i-\rho, q_l+\rho]$ and $\widetilde{p_1} \leq \widetilde{p_2}$. Further, since $(q_i-\rho)-q_x\leq \rho$ and $q_x-(q_l+\rho)\leq\rho$ it holds that $|\widetilde{p_1}-q_x|\leq \rho$. Similarly, $|\widetilde{p_2}-q_y|\leq \rho$. This contradicts the assumption. Hence, $d_{\text{F}}(\langle q_i, \dotso, q_l\rangle, \langle q_i-\rho, q_l+\rho\rangle)\leq \rho$ and $\langle q_i-\rho, q_l+\rho\rangle$ is a forward edge. So, $f_i(q)\geq l$.

    To prove the other direction, assume $f_i(q)\geq l$ and $(P_4( x, j))$ is true for all $  i\leq x\leq  l$. Therefore, by Observation \ref{o: f and b}, it holds that ${d_{\text{F}}(\langle q_ i, \dotso, q_ l\rangle, \langle q_ i-\rho, q_ l+\rho\rangle)\leq \rho}$ and $q_ i-\rho\leq q_ l+\rho$. Let $i\leq x< y\leq l$. Then, there exists points $p_x<p_y$ on the edge $\overline{q_ i-\rho, q_ l+\rho}$ such that $|p_x-q_x|\leq \rho$, $|p_y-q_y|\leq \rho$. Further, there exists points $p_1, p_2\in [s_j, s_{j+1}]$ such that $|p_1-q_x|\leq \rho$, $|p_2-q_y|\leq \rho$ by the properties of predicate $(P_4)$. We define $p_4=\min\{\max\{s_j, p_x\}, s_{j+1}\}$ and $p_6=\min\{\max\{s_j, p_y\}, s_{j+1}\}$. It follows that $p_4\leq p_6$ and $p_4, p_6\in\overline{s_j s_{j+1}}$. Furthermore, if $p_4=s_j$, then $p_x\leq s_j\leq p_1$, resulting in $|s_j-q_x|\leq \rho$. If $p_4=s_{j+1}$ then $p_1\leq s_{j+1}\leq p_x$ and $|s_{j+1}-q_x|\leq \rho$. Therefore, $|p_4-q_x|\leq \rho$. Similarly, it follows that $|p_6-q_y|\leq \rho$. The points $p_4, p_6$ fulfill the conditions of the definition of the monotone in s predicate $(P_6(x, y, j))$, i.e., $(P_6(x, y, j)$ is true.
\end{proof}

Observation \ref{OP1234} and Lemma \ref {lP5} and \ref{lP6} show how we can determine whether a valid sequence of cells is feasible in $F_\rho(q, s)$ using intervals defined by $s$ and $\rho$ for the vertices of $q$ and the forward and backward numbers $f_i(q)$ and $b_i(q)$.

\section{Data Structure}\label{s:Data structure}
In this section, we present two data structures solving the \FQ problem. We start with some assumptions, that can be made for the time series.
Let $s=\langle s_1, \dotso, s_{t}\rangle$ be a time series. Then we can assume that either ${s_{2j-1}\leq s_{2j} \geq s_{2j+1}}$ for all $j$ (\emph{M-shaped}), or  ${s_{2j-1}\geq s_{2j} \leq s_{2j+1}}$ for all $j$ (\emph{W-shaped}), because if $s=\langle s_1, \dotso, s_{2j-1}, s_{2j}, s_{2j+1},\dotso, s_{t}\rangle$ with ${s_{2j-1}\leq s_{2j} \leq s_{2j+1}}$ or ${s_{2j-1}\geq s_{2j} \geq s_{2j+1}}$, it has the same shape as $\langle s_1, \dotso, s_{2j-1},$ $s_{2j+1},\dotso, s_{t}\rangle$. Moreover, we can assume that the complexity of all time series in $S$ is exactly~$t_s$ by simply adding dummy vertices in the end otherwise, since the value of two consecutive vertices can also be equal. In \cref{fig:f and b}, the time series $q$ is W-shaped. 

The query algorithm iterates over all valid sequences of cells $\mathscr C$. By Lemma \ref{lsequenceofcells}, $\mathscr{C}$ is feasible in the free space diagram $F_\rho(q, s)$ if and only if the predicates induced by $\mathscr{C}$ are true. The truth assignment of all needed predicates $(P_1), (P_2), (P_3), (P_4)$ and $(P_5)$ can be determined using intervals defined by $s$ and $\rho$.
Furthermore, $\mathscr{C}$ can only be feasible in $F_\rho(q,s)$ if for all $(i-1, j), (l, j)\in \mathscr{C}$ with $i\leq l$, the monotone in $s$ predicate $(P_6(i, l, j))$ is true. By \cref{lsequenceofcells}, we can use the forward number $f_i(q)$ in the case that $s_j\leq s_{j+1}$ (i.e., $j$ is odd if $s$ is M-shaped) to determine whether $(P_6(i, l, j))$ is true. We define the \emph{forward number} $f_i(\mathscr C)$ as the highest such number $l$ that is needed for $\mathscr C$ to be feasible in $F_\rho(q,s)$. Respectively, if $s_j\geq s_{j+1}$ (i.e., $j$ is even if $s$ is M-shaped) for $b_i(q)$ and we define the \emph{backward number} $b_i(\mathscr C)$. Formally, we get 
\begin{align*}
    f_i(\mathscr{C})=\left\{ 
    \begin{array}{ll}
        l\geq i, & \text{if } \exists\  (i-1, j), (l, j)\in \mathscr{C} \text{ s.t. } j\text{ is odd and }(l+1, j)\notin \mathscr{C}, \\
        i, & \text{otherwise}
    \end{array}\right.
\end{align*}
and 
\begin{align*}
    b_i(\mathscr{C})=\left\{
    \begin{array}{ll}
        l\geq i, & \text{ if }\exists\  (i-1, j), (l, j)\in \mathscr{C}\text{ s.t. }j\text{ is even and }(l+1, j)\notin \mathscr{C}, \\
        i, & \text{otherwise}. 
    \end{array}\right.
\end{align*}
As $\mathscr C$ is valid there exists a unique $j$ such that $(i-1, j), (i, j), \dotso, (l, j)\in \mathscr{C}$. Hence, the numbers $f_i(\mathscr{C})$ and $b_i(\mathscr{C})$ are well-defined.
Note that we do not need $f_1(\mathscr{C})$, $b_1(\mathscr{C})$, $f_{t_q}(\mathscr{C})$ and $b_{t_q}(\mathscr{C})$ because we never consider $(P_6(1, l, j))$ and $(P_6(t_q, l, j)$. 

\subparagraph*{The Data structure.}
    Let $S_M$ be the set of stored time series that are M-shaped and $S_W$ the set of those that are W-shaped. We will describe how $S_M$ is stored. The time series in $S_W$ are stored in the same way after they were mirrored at the origin. Consequently, for those the query algorithm mirrors the query time series $q$ at the origin and is then the same as for the time series in $S_M$.
    
    For all valid sequences of cells $\mathscr{C}$, we build two associated rectangle stabbing data structures storing the time series in $S_M$ as $t_q$-dimensional axis-aligned rectangles. One for the case that the query time series $q$ is M-shaped and the other one for the case that $q$ is W-shaped. Knowing the shape of $q$, Observation \ref{OP1234} and Lemma \ref{lP5} define for every $s\in S_M$ an interval for every vertex $q_i$ of the query time series in which it must lie such that $\mathscr{C}$ can be feasible in $F_{\rho}(q, s)$. For a time series $s$, we store the Cartesian product of those $t_q$ intervals in the associated rectangle stabbing data structure. 
    Note that even if the complexity of the stored time series is greater than $t_q$, we store only a $t_q$-dimensional rectangle for it.

\subparagraph*{The Query Algorithm.}
    Let $q$ be a query time series of complexity $t_q$. The query algorithm starts with computing the numbers $f_1(q), \dotso, f_{t_q}(q), b_1(q), \dotso, b_{t_q}(q)$. For all valid sequences of cells $\mathscr{C}$, we check whether $f_i(\mathscr{C})\leq f_i(q)$ and $b_i(\mathscr{C}) \leq b_i(q)$ for all $i$. If so, we do a query search in the rectangle stabbing data structure depending on $\mathscr C$ and the shape of $q$ with the point $(q_1, q_2, \dotso, q_{t_q})$ and output all time series associated with a rectangle containing this point.

\begin{theorem}\label{t: p rect stabbing}
    The \FQ problem for constant $t_q\geq 2$ and $t_s$ can be solved with a data structure of size $S_R(n, t_q)$ using $Q_R(n, t_q)+\mathcal{O}(k)$ query time, where $k$ is the size of the output (without duplicates) and $S_R(n, t_q)$ denotes the size and $Q_R(n, t_q)$ the query time of a rectangle stabbing data structure that stores $n$ rectangles of dimension $t_q$.
    In particular, there exists a data structure of size $\mathcal{O}(n \log ^{t_q-2} n)$ and query time $\mathcal{O}(\log^{t_q-1}n+k)$ using the rectangle stabbing data structure by Afshani, Arge and Larsen \cite{AAL12}.
\end{theorem}

\begin{proof}
    For $S_M$, two rectangle stabbing data structures are stored for every valid sequence of cells $\mathscr C$. In each, there are stored at most $n$ axis-aligned rectangles of dimension $t_q$. In a valid sequence of cells, every step is either $(i,j), (i, j+1)$ (right) or $(i, j), (i+1,j)$ (upwards) and the first cell is $(1, 1)$ and the last is $(t_q-1, t_s-1)$. Therefore, a valid sequence of cells consists of $t_q+t_s-4$ steps and $t_q-2$ upwards steps. Hence, the number of valid sequence of cells is $\binom{t_q+t_s-4}{t_q-2}$. Since $t_q$ and $t_s$ are considered constant, this is a constant, which completes the proof of the claimed size of the data structure.\footnote{In more detail, if $t_s\geq t_q$ it holds that $\binom{t_q+t_s-4}{t_q-2}\leq \binom{2t_s}{t_q}\leq \left(\frac{2et_s}{t_q}\right)^{t_q}$ by Stirling's approximation of the factorial function.}
    
    Computing the numbers $f_1(q), \dotso, f_{t_q}(q), b_1(q), \dotso, b_{t_q}(q)$ can be done in $\mathcal{O}({t_q}^3 \log t_q)$ time by simply computing all distances $d_F(\langle q_i, \dotso, q_k\rangle, \langle q_i-\rho, q_k+\rho\rangle)$ and $d_F(\langle q_i, \dotso, q_k\rangle,$ $\langle q_i+\rho, q_k-\rho\rangle)$. Each computation takes $\mathcal{O}(t_q \log t_q)$ time by Alt and Godau \cite{CompFreDist}. The query time follows by the fact that for all valid sequences of cells $\mathscr{C}$ we perform at most one query search in an associated rectangle stabbing data structure.
    
    By Observation \ref{OP1234}, Lemma \ref{lsequenceofcells}, \ref{lP5}, and \ref{lP6}, a sequence of cells $\mathscr{C}$ is feasible in $F_\rho(q, s)$ for an M-shaped time series $s\in S$ if and only if 
    all vertices of $q$ lie in the intervals defined by the induced predicates of $\mathscr{C}$ depending on $s$ and Observation \ref{OP1234} and Lemma \ref{lP5}, and
    $f_i(q)\geq f_i(\mathscr{C})$ and $b_i(q)\geq b_i(\mathscr{C})$ for all $i\in \{2, \dotso, t_q-1\}$.
     Therefore, the correctness follows by the fact that we iterated over all valid sequences of cells and by Lemma \ref{lAltGodau}.
\end{proof}

The output in Theorem~\ref{t: p rect stabbing} may contain a constant fraction of duplicates. As such it cannot be easily used for range counting. To remove duplicates, one can use standard techniques, such as hashing.

Using an orthogonal range searching data structure it is possible to store the time series as $t_s$-dimensional points and the query time series defines then $t_s$-dimensional axis-aligned rectangles. 

\begin{corollary}
    The \FQ problem for constant $t_q$ and $t_s>2$ can be solved with a data structure of size $S(n, t_s)$ using $Q(n, t_s)+\mathcal{O}(k)$ query time, where $k$ is the size of the output (without duplicates) and $S(n, t_s)$ denotes the size and $Q(n, t_s)$ the query time of an orthogonal range searching data structure that stores $n$ points in dimension $t_s$. In particular, there exists a data structure of size $\mathcal{O}\left(n(\log n/\log \log n)^{t_s-1}\right)$ and query time $\mathcal{O}(\log n(\log n/\log \log n)^{t_s-3}+k)$. 
\end{corollary}
\begin{proof}
    We use a similar idea as in the proof of Theorem \ref{t: p rect stabbing} with the difference that the time series in $S$ are stored as $t_s$-dimensional points and the query time series defines $t_s$-dimensional axis-aligned rectangles. We build two data structures one for M-shaped query time series and one for W-shaped time series. We describe only the one for the M-shaped case here. The other one is build symmetrically. Note that in the following we exchange the role of $q$ and $s$ and consider $F_\rho(s, q)$ instead of $F_\rho(q, s)$. For every valid sequence of cells~$\mathscr{C}$ in $F_\rho(s,q)$, we build an orthogonal range searching data structures storing the time series $s$ where $f_j(s)\geq f_j(\mathscr{C})$ and $b_j(s)\geq b_j(\mathscr{C})$ for all $j$. The query algorithm computes for all valid sequences of cells $\mathscr{C}$ a rectangle $R(\mathscr{C})$ such that $\mathscr C$ is feasible in $F_\rho(s, q)$ if and only if  $(s_1, \dotso, s_{t_s})\in R(\mathscr C)$ and $s$ is stored in the data structure defined by $\mathscr C$. The rectangle can be computed with Observation \ref{OP1234} and Lemma \ref{lP5}. The correctness follows by Observation \ref{OP1234} and Lemma \ref{lsequenceofcells}, \ref{lP5} and \ref{lP6}. The bounds for the size and query time follow in the same way as in Theorem \ref{t: p rect stabbing} and by using the orthogonal range searching data structure by Afshani, Arge and Larsen \cite{AAL12}.
\end{proof}

We believe that with some modifications it is possible to solve the \FQ problem also for the case where the complexity of the query time series is not given at preprocessing time within the same bounds. Further, we believe that using orthogonal intersection searching (see~\cite{EM1981}) it is possible to build a data structure of size $\mathcal{O}(n\log^{t_q}n)$ and query time $\mathcal{O}(\log^{t_q-1}n)$ for the \FQ problem where the distance threshold is not given at preprocessing time.

\bibliography{Literature.bib}

\begin{thebibliography}{10}

\bibitem{AAL12}
P.~Afshani, L.~Arge, and K.G. Larsen.
\newblock Higher-dimensional orthogonal range reporting and rectangle stabbing
  in the pointer machine model.
\newblock In {\em Proceedings of the 2012 Symposuim on Computational Geometry},
  pages 323--338, 2012.
\newblock \href {https://doi.org/10.1145/2261250.2261299}
  {\path{doi:10.1145/2261250.2261299}}.

\bibitem{AfshaniDriemel18}
P.~Afshani and A.~Driemel.
\newblock On the complexity of range searching among curves.
\newblock In {\em Proceedings of the 2018 Annual ACM-SIAM Symposium on Discrete
  Algorithms (SODA)}, pages 898--917, 2018.
\newblock \href {https://doi.org/10.1137/1.9781611975031.58}
  {\path{doi:10.1137/1.9781611975031.58}}.

\bibitem{Alt2009}
H.~Alt.
\newblock The computational geometry of comparing shapes.
\newblock In {\em Efficient Algorithms: Essays Dedicated to Kurt Mehlhorn on
  the Occasion of His 60th Birthday}, pages 235--248. Springer Berlin
  Heidelberg, 2009.
\newblock \href {https://doi.org/10.1007/978-3-642-03456-5_16}
  {\path{doi:10.1007/978-3-642-03456-5_16}}.

\bibitem{CompFreDist}
H.~Alt and M.~Godau.
\newblock Computing the {F}r\'{e}chet distance between two polygonal curves.
\newblock {\em International Journal of Computational Geometry and
  Applications}, 5(01\& 02):75--91, 1995.
\newblock \href {https://doi.org/10.1142/S0218195995000064}
  {\path{doi:10.1142/S0218195995000064}}.

\bibitem{BDNP21}
K.~Bringmann, A.~Driemel, A.~Nusser, and I.~Psarros.
\newblock Tight bounds for approximate near neighbor searching for time series
  under the {F}r\'{e}chet distance.
\newblock In {\em Proceedings of the 2022 Annual {ACM}-{SIAM} Symposium on
  Discrete Algorithms ({SODA})}, pages 517--550, 2022.
\newblock \href {https://doi.org/10.1137/1.9781611977073.25}
  {\path{doi:10.1137/1.9781611977073.25}}.

\bibitem{Cha86}
B.~Chazelle.
\newblock Filtering search: a new approach to query-answering.
\newblock {\em SIAM Journal on Computing}, 15(03):703--724, 1986.
\newblock \href {https://doi.org/10.1137/0215051} {\path{doi:10.1137/0215051}}.

\bibitem{Chazelle_1990}
B.~Chazelle.
\newblock Lower bounds for orthogonal range searching: I. the reporting case.
\newblock {\em Journal of the {ACM}}, 37(02):200--212, 1990.
\newblock \href {https://doi.org/10.1145/77600.77614}
  {\path{doi:10.1145/77600.77614}}.

\bibitem{Cheng2023SolvingFD}
Siu-Wing Cheng and Haoqiang Huang.
\newblock Solving {F}r\'{e}chet distance problems by algebraic geometric
  methods.
\newblock {\em ArXiv}, abs/2308.14569, 2023.
\newblock \href {https://doi.org/10.48550/arXiv.2308.14569}
  {\path{doi:10.48550/arXiv.2308.14569}}.

\bibitem{DEBERG2013747}
Mark {de Berg}, Atlas~F. Cook, and Joachim Gudmundsson.
\newblock Fast {F}r\'{e}chet queries.
\newblock {\em Computational Geometry}, 46(6):747--755, 2013.
\newblock \href {https://doi.org/10.1016/j.comgeo.2012.11.006}
  {\path{doi:10.1016/j.comgeo.2012.11.006}}.

\bibitem{DP21}
A.~Driemel and I.~Psarros.
\newblock {ANN} for time series under the {F}r\'{e}chet distance.
\newblock In {\em Algorithms and Data Structures}, pages 315--328, 2021.
\newblock \href {https://doi.org/10.1007/978-3-030-83508-8_23}
  {\path{doi:10.1007/978-3-030-83508-8_23}}.

\bibitem{DS17}
A.~Driemel and F.~Silvestri.
\newblock Locality-sensitive hashing of curves.
\newblock In {\em 33rd International Symposium on Computational Geometry},
  volume~77, pages 37:1--37:16, 2017.
\newblock \href {https://doi.org/10.4230/LIPIcs.SoCG.2017.37}
  {\path{doi:10.4230/LIPIcs.SoCG.2017.37}}.

\bibitem{EM1981}
H.~Edelsbrunner and H.A. Maurer.
\newblock On the intersection of orthogonal objects.
\newblock {\em Information Processing Letters}, 13(04):177--181, 1981.
\newblock \href {https://doi.org/10.1016/0020-0190(81)90053-3}
  {\path{doi:10.1016/0020-0190(81)90053-3}}.

\bibitem{FFK20}
A.~Filtser, O.~Filtser, and M.J. Katz.
\newblock Approximate nearest neighbor for curves: simple, efficient, and
  deterministic.
\newblock {\em Algorithmica}, 2022.
\newblock \href {https://doi.org/10.1007/s00453-022-01080-1}
  {\path{doi:10.1007/s00453-022-01080-1}}.

\bibitem{CompBio}
M.~Jiang and B.~Zhu Y.~Xu.
\newblock Protein structure-structure alignment with discrete {F}r\'{e}chet
  distance.
\newblock {\em Journal of Bioinformatics and Computational Biology},
  06(01):51--64, 2008.
\newblock \href {https://doi.org/10.1142/s0219720008003278}
  {\path{doi:10.1142/s0219720008003278}}.

\bibitem{Meulemans2014SimilarityMA}
W.~Meulemans.
\newblock {\em Similarity measures and algorithms for cartographic
  schematization}.
\newblock PhD thesis, Technische Universiteit Eindhoven, 2014.
\newblock \href {https://doi.org/10.6100/IR777493}
  {\path{doi:10.6100/IR777493}}.

\bibitem{Sriraghavendra2007FrchetDB}
E.~Sriraghavendra, K.~Karthik, and C.~Bhattacharyya.
\newblock {F}r\'{e}chet distance based approach for searching online
  handwritten documents.
\newblock In {\em Ninth International Conference on Document Analysis and
  Recognition (ICDAR 2007)}, volume~1, pages 461--465, 2007.
\newblock \href {https://doi.org/10.1109/ICDAR.2007.4378752}
  {\path{doi:10.1109/ICDAR.2007.4378752}}.

\bibitem{GIS_3}
K.~Toohey and M.~Duckham.
\newblock Trajectory similarity measures.
\newblock {\em SIGSPATIAL Special}, 7(1):43--50, 2015.
\newblock \href {https://doi.org/10.1145/2782759.2782767}
  {\path{doi:10.1145/2782759.2782767}}.

\bibitem{Wenk2006AddressingTN}
C.~Wenk, R.~Salas, and D.~Pfoser.
\newblock Addressing the need for map-matching speed: localizing global
  curve-matching algorithms.
\newblock In {\em 18th International Conference on Scientific and Statistical
  Database Management (SSDBM'06)}, pages 379--388, 2006.
\newblock \href {https://doi.org/10.1109/SSDBM.2006.11}
  {\path{doi:10.1109/SSDBM.2006.11}}.

\bibitem{Wylie2013ProteinCP}
T.~Wylie and B.~Zhu.
\newblock Protein chain pair simplification under the discrete {F}r\'{e}chet
  distance.
\newblock {\em IEEE/ACM Transactions on Computational Biology and
  Bioinformatics}, 10(6):1372--1383, 2013.
\newblock \href {https://doi.org/10.1109/TCBB.2013.17}
  {\path{doi:10.1109/TCBB.2013.17}}.

\bibitem{ZHU2023619}
Y.~Zhu, J.~Peng, H.~Liu, and Z.~Lan.
\newblock Chapter 26 - {A}nalysis of nonadiabatic molecular dynamics
  trajectories.
\newblock In {\em Quantum Chemistry in the Age of Machine Learning}, pages
  619--651. Elsevier, 2023.
\newblock \href {https://doi.org/10.1016/B978-0-323-90049-2.00013-5}
  {\path{doi:10.1016/B978-0-323-90049-2.00013-5}}.

\end{thebibliography}

\end{document}